\documentclass[letterpaper, 10 pt, conference]{ieeeconf}

     \IEEEoverridecommandlockouts 
    \makeatletter

    \let\proof\@undefined
    \let\endproof\@undefined
    \makeatother
    \let\labelindent\relax

    \usepackage{jhoagg}
    \usepackage{mathtools}
    \usepackage{amsfonts}
    \usepackage{amssymb,latexsym}
    \usepackage[psamsfonts]{eucal}
    \usepackage{amsthm}
    \usepackage{graphicx}
    \usepackage[inline]{enumitem}
    \usepackage{indentfirst}
    \usepackage{setspace}
    \usepackage{microtype}
    \usepackage{threeparttable}
    \usepackage{float}
    \usepackage{capt-of}
    \usepackage{hyperref}
    \usepackage{cleveref}
    \usepackage{afterpage}
    \usepackage{placeins}
    \usepackage{cancel}
    \usepackage{cite}
    \usepackage{leftidx}
    \usepackage{breqn}
    \usepackage[export]{adjustbox}
    \usepackage{comment}
    \usepackage[linesnumbered,ruled,vlined]{algorithm2e}
    \usepackage[dvipsnames]{xcolor}
    \usepackage{soul}
    \usepackage{siunitx}
    \usepackage{arydshln}

    \usepackage{tikz}

    \let\originalleft\left
    \let\originalright\right
    \renewcommand{\left}{\mathopen{}\mathclose\bgroup\originalleft}
    \renewcommand{\right}{\aftergroup\egroup\originalright}

    \AtBeginDocument{%
      \setlength{\abovedisplayskip}{3.0pt plus 2pt minus 3pt}
      \setlength{\belowdisplayskip}{3.0pt plus 2pt minus 3pt}
      \setlength{\abovedisplayshortskip}{0pt plus 2pt}
      \setlength{\belowdisplayshortskip}{2.0pt plus 2pt minus 2pt}
    }

    \makeatletter
    \def\thm@space@setup{%
      \thm@preskip=3pt plus 1pt minus 1pt
      \thm@postskip=\thm@preskip
    }
    \makeatother
    
    \usepackage{etoolbox}
    \makeatletter
    \expandafter\patchcmd\csname\string\proof\endcsname{\topsep6\p@\@plus6\p@}{\topsep3\p@\@plus3\p@}{}{\PackageWarning{patch}{proof spacing patch FAILED}}
    \makeatother

    \makeatletter
    \patchcmd{\thebibliography}%
      {\vskip 0.3\baselineskip plus 0.1\baselineskip minus 0.1\baselineskip}%
      {\vskip 0pt}{}{\PackageWarning{patch}{bib vskip patch FAILED}}
    \patchcmd{\thebibliography}%
      {\itemsep 0pt plus .5pt\relax}%
      {\itemsep 0pt plus .5pt\relax \topsep 0pt\relax \parskip 0pt\relax}%
      {}{\PackageWarning{patch}{bib topsep patch FAILED}}
    \makeatother

    \newcounter{thm}
    \newtheorem{theorem}[thm]{\indent Theorem}
    \newtheorem{assumption}{\indent Assumption}
    \newtheorem{proposition}{\indent Proposition}
    \newtheorem{lemma}{\indent Lemma}
    \newtheorem{corollary}{\indent Corollary}
    \newtheorem{definition}{\indent Definition}
    \newtheorem{example}{\indent Example}

    \newtheorem{problem}{\indent Problem}
    \allowdisplaybreaks

    \renewcommand{\theenumi}{{\it (\alph{enumi})}}
    \renewcommand{\labelenumi}{\theenumi}

    		\newcommand\xqed[1]{%
      \leavevmode\unskip\penalty9999 \hbox{}\nobreak\hfill
      \quad\hbox{#1}}
    \newcommand\exampletriangle{\xqed{$\triangle$}}

    \usepackage{accents}

    \newlength\figureheight
    \newlength\figurewidth

    \allowdisplaybreaks

    \graphicspath{ {Figures/} }
    \DeclareGraphicsExtensions{.pdf,.png,.jpg}

    \DeclareMathAlphabet{\mathcal}{OMS}{cmsy}{m}{n} 

    \crefname{equation}{}{}

    \SetKwInput{KwInit}{Initialize}
    \SetKwFunction{SafeSet}{SafeSet}
    \SetKwFunction{GetSafe}{GetSafe}
    \SetKwFunction{GetUnsafe}{GetUnsafe}
    
    \SetCommentSty{mycommfont}

\begin{document}

\title{Predicted-Flow Control Barrier Functions for Non-Control-Affine Systems}

\author{Amirsaeid Safari and Jesse B. Hoagg\vspace{-3ex}
\thanks{A. Safari and J. B. Hoagg are with the Department of Mechanical and Aerospace Engineering, University of Kentucky, Lexington, KY, USA. (e-mail: amirsaeid.safari@uky.edu, jesse.hoagg@uky.edu).}
\thanks{Supported in part by NSF (2450718) and USDA (2024-69014-42393).}
}
\maketitle

\begin{abstract}
Control barrier functions (CBFs) enforce safety through conditions imposed pointwise in time without consideration of state evolution over a future horizon.
Thus, CBF-based controls are typically myopic. 
Predicted-flow CBFs (P-CBFs) generalize CBFs from a function of the current state to a functional of the predicted flow under a parametrized control.
P-CBFs can certify safety over the entire prediction horizon while simultaneously allowing for performance optimization over the horizon. 
However, prior work with P-CBFs only applies to control-affine systems, and suffers from a limitation where the prediction horizon can shrink or even vanish.
This article addresses both of these shortcomings by introducing a planning control that smoothly transitions from a parametric plan to a backup control (if needed) and a time-shift parameter that determines if transition is needed. 
The evolution of the control-plan and time-shift parameters is determined from a single guaranteed-feasible convex optimization, which reduces to a quadratic program (QP) if the control limits are a convex polyhedron. 
The real-time executed control is determined from the instantaneous planning control and time shift.
This method simultaneously addresses safety certification and performance optimization over the fixed prediction horizon. 
The approach is compared to nonlinear model predictive control in simulation of an autonomous car navigating a dense obstacle environment.
\end{abstract}

\vspace{-10pt}

\section{Introduction}

Control barrier functions (CBFs) can be used to synthesize controls that enforce forward invariance inside a safe set $\SC_\rms$ \cite{blanchini1999,wieland2007}.
For a control-affine system, a quadratic program (QP) can be used to compute controls that achieve forward invariance \cite{ames2016control}.
However, synthesizing a valid CBF on a large subset of $\SC_\rms$ is challenging due to control input constraints \cite{choi2021hj}; existing synthesis methods are typically limited to low-dimensional systems or restricted dynamics (e.g., \cite{choi2021hj,wang2018}).
Furthermore, pointwise CBF-based optimization is myopic because safety is enforced at each time instant based on the current state without consideration of how the trajectory will evolve over a future time horizon \cite{breeden2022,agrawal2025gatekeeper}.

Backup-CBF methods address the validity challenge by relying on a backup controller $u_\rmb$ that makes a backup set $\SC_\rmb \subset \SC_\rms$ forward invariant \cite{gurriet2020,backupautomatic}.
Safety is enforced through an implicit control-forward-invariant subset of $\SC_\rms$ where the state prediction under $u_\rmb$ stays in $\SC_\rms$ and reaches $\SC_\rmb$ in finite time. 
This implicit control-forward-invariant set is often small relative to $\SC_\rms$.
Conservatism can be reduced by delaying the switch to $u_\rmb$ 
\cite{agrawal2025gatekeeper,singletary2022safe}.
Nevertheless, backup-CBF methods use the prediction horizon only for safety and not for performance optimization.
Thus, the resulting controls are myopic in the sense that they do not optimize cost over a prediction horizon.

The myopic nature of pointwise CBF optimization can be addressed with layered architectures that combine a planner with a CBF safety filter \cite{grandia2021layered,sforni2024receding}; however, CBF validity with input constraints is not addressed and multiple optimizations are required. 
Other methods encode safety and performance over a horizon in a single optimization \cite{breeden2022,vahs2024}.
However, validity/feasibility with input constraints is not guaranteed.

Predicted-flow CBFs (P-CBFs) can be used to address non-myopic performance and safety with guaranteed feasibility \cite{safari2026flowbarrier}.
P-CBFs generalize CBFs from a function of the current state to a functional of the predicted flow under a control plan, which is parametrized by a finite-dimensional variable over a prediction horizon.
The method in \cite{safari2026flowbarrier} uses a terminal constraint to ensure that the predicted flow ends in a known control-forward-invariant subset of $\SC_\rms$, which is analogous to the terminal constraint in model predictive control \cite{mayne2000}.
A time-shift parameter modulates the prediction horizon to ensure feasibility.
One drawback of \cite{safari2026flowbarrier} is that the prediction horizon can shrink or even vanish if needed for feasibility. 
Thus, \cite{safari2026flowbarrier} uses a backup control as the final failsafe, and a separate mechanism is used to recover a nonempty prediction horizon. 
Also, \cite{safari2026flowbarrier} is restricted to control-affine systems.

This article advances P-CBF methods by overcoming the shrinking-prediction-horizon limitation of \cite{safari2026flowbarrier} and addressing systems that are not necessarily affine in the control. 
These advances are achieved by introducing a planning control that smoothly transitions from a parametric plan to a backup controller only if needed.
Time-shift parameter determines if transition is needed; however, the prediction horizon remains fixed. 
The rates of change of the control-plan and time-shift parameters are determined from a single guaranteed-feasible convex optimization, which reduces to a QP if the control constraints are a convex polyhedron.
The real-time executed control is determined from the instantaneous planning control and time shift.
This work proves the candidate P-CBFs are valid, and thus simultaneously addresses safety certification and performance optimization over a fixed prediction horizon. 


\section{Preliminaries}\label{sec:background}

Let $\SD \subseteq \BBR^n$, and let $\mu \colon \SD \to \BBR$ be continuous.
The \textit{radial cone} $R_\SD \colon \SD \rightrightarrows \BBR^n$ is defined by
\begin{align*}
R_\SD(x) &\triangleq \{ \nu \in \BBR^n \colon \exists \,  \varepsilon > 0 \text{ s.t. } \forall\, s \in (0,\varepsilon),\, x + s\nu \in \SD \}.
\end{align*}
Note that if $x \in \operatorname{int}\SD$, then $R_\SD(x) = \BBR^n$.
The function $\mu$ is \textit{right-side directionally differentiable on $\SD$} if for all $(x,\nu) \in \SD \times R_\SD(x)$,
\begin{equation*}
D_{\nu}  \mu(x) \triangleq \lim_{s \downarrow 0} \frac{\mu(x + s \nu) - \mu(x)}{s}
\end{equation*}
exists.
If $\mu$ is differentiable on $\SD$, then $D_{\nu} \mu(x) = L_{\nu} \mu(x)$, where $L_\nu \mu(x) \triangleq \mu^\prime(x) \nu$ is the Lie derivative of $\mu$ along $\nu$.
For brevity, we omit ``right-side'' for the remainder of this article.
The next result concerns the time derivative of $\mu(y(t))$ from the right side. See \cite[Lemma~1]{safari2026flowbarrier} for a proof.

\begin{lemma}\rm\label{lem:dini_directional}
Assume $\mu$ is locally Lipschitz and directionally differentiable on $\SD$.
Let $y \colon [0,\infty) \to \SD$ be differentiable such that for all $t \geq 0$, $\dot{y}(t) \in R_\SD(y(t))$.
Then, for all $t \geq 0$,
\begin{equation*}
\frac{\rmd^+}{\rmd t} \mu(y(t)) \triangleq \lim_{s \downarrow 0} \frac{\mu(y(t+s)) - \mu(y(t))}{s} = D_{\dot{y}(t)} \mu(y(t))
\end{equation*}
exists. 
\end{lemma}


\section{Predicted-Flow Control Barrier Functions}\label{sec:pcbf}

Consider the dynamics
\begin{equation}\label{eq:system}
\dot x(t) = f(x(t), u(t)),
\end{equation}
where $f\colon \BBR^n \times \SU \to \BBR^n$ is continuously differentiable,
$x(t) \in \BBR^n$ is the state,
$x(0) = x_0 \in \BBR^n$ is the initial condition, and
$u(t) \in \SU \subseteq \BBR^m$ is the control.
The control $u$ is \textit{admissible} if for all $t\ge0$, $u(t) \in \SU$. 
We assume \eqref{eq:system} is forward complete under all bounded and continuous $u$.

Consider the \textit{planning control} $\pi(\cdot, \cdot\,; \theta) \colon [0, \infty) \times \BBR^n \to \BBR^m$, which is parameterized by $\theta \in \BBR^d$.
The planning control $\pi$ is continuous on $[0,\infty) \times \BBR^n \times \BBR^d$, and for all $t \ge 0$, $\pi(t,\cdot\,;\cdot)$ is continuously differentiable on $\BBR^n \times \BBR^d$.

Next, let $k \colon \BBR^d \to \BBR$ be continuously differentiable, and define the \textit{admissible parameter set}
\begin{equation*} 
\Theta \triangleq \{ \theta \in \BBR^d \colon k(\theta) \ge 0 \},
\end{equation*}
where $\Theta$ is such that for all $(t, x,\theta) \in [0,\infty) \times \BBR^n \times \Theta$, $\pi(t, x; \theta) \in \SU$.
In other words, $\theta \in \Theta$ is a parameterization of planning controls that are admissible. 
The next example constructs a time-varying planning control that is independent of the state. 
Later in this article, we present and use a time-and-state-dependent planning control.

\begin{example}\rm\label{ex:foh}
Let $\beta_1,\ldots,\beta_p \colon [0,\infty) \to [0,\infty)$ be continuous such that for all $t \ge 0$, $\sum_{i=1}^{p} \beta_i(t) = 1$.
Degree-one B-splines are one choice for $\beta_1,\ldots,\beta_p$; see Figure~\ref{fig:basis_functions}.
Define $\kappa(t;\theta) = \sum_{i=1}^{p} \theta_i \, \beta_i(t)$, where $\theta_i \in \BBR^m$, $\theta = [\theta_1^\top \;\theta_2^\top \;\ldots \;\theta_p^\top]^\top \in \BBR^{d}$, and $d = pm$.

Let $\pi = \kappa$.
Since $\pi$ is a convex combination of $\theta_1,\ldots,\theta_p$, selecting $k$ such that $\Theta \subseteq \SU^p$ implies that for all $(t,\theta) \in [0,\infty) \times \Theta$, $\pi(t;\theta) \in \SU$.
Section~\ref{sec:qp} provides a simple construction of $k$ if $\SU$ is a convex polyhedron. 
\exampletriangle
\end{example}

Next, let $\gamma \geq 0$ be the \textit{time shift}, and let the \textit{predicted flow} $\phi(\cdot; x, \theta, \gamma) \colon [\gamma, \infty) \to \BBR^n$ satisfy
\begin{equation}\label{eq:flow_def}
\phi(\tau; x, \theta, \gamma) \!=\! x + \! \int_{\gamma}^{\tau} \!\! f(\phi(\sigma; x, \theta, \gamma), \pi(\sigma, \phi(\sigma; x, \theta, \gamma); \theta)) \, \rmd\sigma.
\end{equation}
Differentiating \eqref{eq:flow_def} with respect to $\tau$ yields
\begin{equation}\label{eq:flow_velocity}
\frac{\partial \phi}{\partial \tau}(\tau; x, \theta, \gamma) = f(\phi(\tau; x, \theta, \gamma),\, \pi(\tau, \phi(\tau; x, \theta, \gamma); \theta)),
\end{equation}
which is the evolution of $\phi$  given $(x,\theta,\gamma)$ from initial condition $\phi(\gamma; x, \theta, \gamma) \!=\! x$. 
Hence, the predicted flow $\phi(\tau;x,\theta,\gamma)$ is the solution to \eqref{eq:system} at $t = \tau-\gamma$ with $x_0 = \phi(\gamma; x, \theta,\gamma)$ and $u(t) = \pi(t+\gamma, \phi(t+\gamma; x, \theta, \gamma); \theta)$.
Note that $\gamma$ is a shift in the temporal argument of $\pi$ used to obtain the predicted flow.
For all real time $t \ge 0$, it follows that $\phi(\tau; x(t), \theta, \gamma)$ is a trajectory prediction of \eqref{eq:system} from its current value $x(t)$ under the control $\pi(\tau, \phi(\tau; x(t), \theta, \gamma); \theta)$ at prediction time $\tau-\gamma$.

To influence the time evolution of $\theta$ and $\gamma$, we let $\theta : [0,\infty) \to \BBR^d$ and $\gamma \colon [0,\infty) \to \BBR$ be the solutions to
\begin{equation}\label{eq:theta_gamma_dyn}
\dot{\theta}(t) = \omega(t), \qquad \dot{\gamma}(t) = z(t),
\end{equation}
where $\theta(0) = \theta_0 \in \BBR^d$, $\gamma(0) = \gamma_0 \geq 0$, $\omega : [0,\infty) \to \Omega \subseteq \BBR^d$ and $z \colon [0,\infty) \to \SZ \subseteq \BBR$ are control inputs.

\begin{figure}[t]
\centering
\includegraphics[width=1.0\columnwidth, clip=true, trim=0.0in 0.15in 0.0in 0.0in]{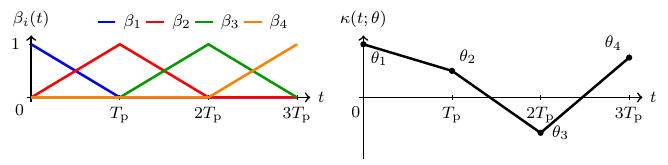}
\caption{Piecewise-linear $\kappa$ constructed from degree-one B-splines $\beta_1,\!\ldots,\!\beta_p$.}
\label{fig:basis_functions}
\end{figure}

Next, let $u$ be given by 
\begin{equation}\label{eq:u_follows_plan}
u(t) = \pi(\gamma(t), x(t); \theta(t)).
\end{equation}
To understand \eqref{eq:u_follows_plan}, consider controls $\omega(t) \equiv 0$ and $z(t) \equiv 1$, where $\gamma_0 = 0$. 
In this case, $\theta(t) = \theta_0$, $\gamma(t) = t$, and $u(t) = \pi(t, x(t); \theta_0)$, which implies that the state trajectory equals the predicted flow, that is, $x(t) = \phi(t; x_0, \theta_0, 0)$.
In general, $\omega$ modifies the parameter $\theta$ and $z$ modifies the temporal argument of the planning control used for execution in \eqref{eq:u_follows_plan}.
We note that $z(t)>0$ shifts the executed control \eqref{eq:u_follows_plan} later in the plan, whereas $z(t)<0$ shifts it earlier in the plan.

It follows that \cref{eq:system,eq:theta_gamma_dyn,eq:u_follows_plan} can be written as
\begin{equation}\label{eq:augmented}
\dot{\bar{x}}(t) = \bar{f}(\bar{x}(t)) + \bar{g} \, v(t),
\end{equation}
where $\bar n \triangleq n+d+1$, and 
 \begin{gather}
\bar x \triangleq [x^\top \; \theta^\top \; \gamma]^\top, \,\,
\bar x_0 \triangleq [x_0^\top \; \theta_0^\top \; \gamma_0]^\top, \,\,
v \triangleq [\omega^\top \;\; z]^\top, \label{eq:augmented.2}\\
                 \bar{f}(\bar{x}) \triangleq \begin{bmatrix}
                       f(x,\pi(\gamma, x;\theta))\\
                       0_{(d+1) \times 1}
                 \end{bmatrix},\quad
\bar{g} \triangleq \begin{bmatrix}
                   0_{n \times (d+1)} \\
                   I_{d+1}
           \end{bmatrix}. \label{eq:augmented.3}
\end{gather}
Hence, \cref{eq:system} combined with \cref{eq:theta_gamma_dyn,eq:u_follows_plan} is affine in the new control $v$, which determines the rate of change of $\theta$ and $\gamma$.

Let $T > 0$ be the prediction horizon, and for all $i \in \{1,\ldots,\ell\}$, let $H_i \colon C([\gamma, \gamma+T], \BBR^n) \to \BBR$ be a functional such that
\begin{equation}\label{eq:psi_functional}
\psi_i(\bar{x}) \triangleq H_i[\phi(\cdot; \bar{x})]
\end{equation}
is locally Lipschitz and directionally differentiable on
\begin{equation*} 
\SH \triangleq \{\bar{x} \in \BBR^{\bar n} \colon  k(\theta) \geq 0, \gamma \geq 0, \psi_1(\bar{x}) \geq 0, \ldots, \psi_\ell(\bar{x}) \geq 0 \}.
\end{equation*}
Note that $\SH$ is the set of $\bar x$ such that $\theta \in \Theta$, $\gamma \geq 0$, and each functional of the predicted flow $\phi$ is nonnegative.
The prediction window $[\gamma, \gamma + T]$ has fixed length $T$ for all $\gamma \geq 0$.
Next, we introduce the concept of a predicted-flow CBF.

\begin{definition}\label{def:pcbf}
\rm
Assume $\psi_1,\!\ldots,\!\psi_\ell$ given by \eqref{eq:psi_functional} are locally Lipschitz and directionally differentiable on $\SH$.
Then, $(\psi_1,\!\ldots,\!\psi_\ell)$ is a \textit{predicted-flow control barrier function (P-CBF) \!$\ell$-tuple for \eqref{eq:augmented} on $\SH$ given $\pi$ and $k$} if there exist extended class-$\SK$ functions $\alpha_1,\!\ldots,\!\alpha_\ell,\!\beta_\theta,\!\beta_\gamma$ such that for all $\bar{x} \!\in\! \SH$, $K_\SH(\bar{x})$ is nonempty, where $K_\SH \colon \SH \!\rightrightarrows\! \Omega \!\times\! \SZ$ is defined by
\begin{align}\label{eq:pfcbf_K}
K_\SH(\bar{x}) \!\triangleq\! \{&\hat{v} \!\in\! \Omega \!\times\! \SZ \colon k^\prime(\theta)\hat{\omega} \!+\! \beta_\theta(k(\theta)) \!\geq\! 0, \; \hat{z} \!+\! \beta_\gamma(\gamma) \!\geq\! 0, \nn\\
& D_{\bar{f}(\bar{x}) + \bar{g}\hat{v}} \psi_1(\bar{x}) + \alpha_1(\psi_1(\bar{x})) \geq 0, \ldots, \nn\\
& D_{\bar{f}(\bar{x}) + \bar{g}\hat{v}} \psi_\ell(\bar{x}) + \alpha_\ell(\psi_\ell(\bar{x})) \geq 0 \}.
\end{align}
\end{definition}

The control $z$ modifies the temporal argument $\gamma$ in the executed control \eqref{eq:u_follows_plan}, which provides an additional degree of freedom that can help ensure \eqref{eq:pfcbf_K} is nonempty.
Section~\ref{sec:sofc} takes advantage of this to construct a valid P-CBF pair.

The next result shows that if $(\psi_1,\ldots,\psi_\ell)$ is a P-CBF $\ell$-tuple on $\SH$, then any control selected pointwise from $K_\SH(\bar{x})$ makes $\SH$ forward invariant.

\begin{theorem}\rm\label{thm:pcbf_invariance}
Assume $(\psi_1,\ldots,\psi_\ell)$ is a P-CBF $\ell$-tuple for \eqref{eq:augmented} on $\SH$.
Let $v_{\rm fi} \colon \SH \to \Omega \times \SZ$ be such that for all $\bar{x} \in \SH$, $v_{\rm fi}(\bar{x}) \in K_\SH(\bar{x})$.
Let $\bar{x}_0 \in \SH$, and assume there exists $t_\rmm \in (0, \infty]$ such that \eqref{eq:augmented} with $v = v_{\rm fi}$ has a unique solution on $[0, t_\rmm)$.
Then, for all $t \in [0, t_\rmm)$, $\bar{x}(t) \in \SH$.
\end{theorem}

\begin{proof}
For $i \in \{1,\ldots,\ell\}$, let $\eta_i \colon [0,\infty) \to \BBR$ satisfy $\dot{\eta}_i = -\alpha_i(\eta_i)$, where $\eta_i(0) = \psi_i(\bar{x}_0)$.
Since $\psi_i(\bar{x}_0) \ge 0$ and $\alpha_i$ is extended class-$\SK$, it follows that for all $t \geq 0$, $\eta_i(t) \geq 0$.
Similarly, let $\eta_\theta,\eta_\gamma \colon [0,\infty) \to \BBR$ satisfy $\dot{\eta}_\theta = -\beta_\theta(\eta_\theta)$ and $\dot{\eta}_\gamma = -\beta_\gamma(\eta_\gamma)$, where $\eta_\theta(0) = k(\theta_0)$ and $\eta_\gamma(0) = \gamma_0$.
Since $k(\theta_0) \ge 0$, $\gamma_0 \ge 0$, and $\beta_\theta$ and $\beta_\gamma$ are extended class-$\SK$, it follows that for all $t \geq 0$, $\eta_\theta(t) \geq 0$ and $\eta_\gamma(t) \geq 0$.

Since $(\psi_1,\ldots,\psi_\ell)$ is a P-CBF $\ell$-tuple on $\SH$, Definition~\ref{def:pcbf} implies that for all $\bar{x} \in \SH$, $K_\SH(\bar{x})$ is nonempty.
Thus, Lemma~\ref{lem:dini_directional} and \eqref{eq:pfcbf_K} imply that for all $t \in [0, t_\rmm)$,
\begin{gather*}
\frac{\rmd}{\rmd t} k(\theta(t)) \geq -\beta_\theta(k(\theta(t))), \quad
\frac{\rmd}{\rmd t} \gamma(t) \geq -\beta_\gamma(\gamma(t)), \\
\frac{\rmd^+}{\rmd t} \psi_i(\bar{x}(t)) \geq -\alpha_i(\psi_i(\bar{x}(t))), \quad \text { for } i \in \{1,\ldots,\ell\}.
\end{gather*}
Thus, \cite[Lemma~3.4]{khalil2002nonlinear} implies that for all $t \in [0, t_\rmm)$, $\psi_i(\bar{x}(t)) \geq \eta_i(t) \geq 0$, $k(\theta(t)) \geq \eta_\theta(t) \geq 0$, and $\gamma(t) \geq \eta_\gamma(t) \geq 0$.
Thus, for all $t \in [0, t_\rmm)$, $\bar{x}(t) \in \SH$.
\end{proof}

\section{Problem Formulation}\label{sec:problem}

The remainder of this article considers the non-control-affine dynamics \cref{eq:system} with control \eqref{eq:u_follows_plan}, where $\theta$ and $\gamma$ evolve according to \cref{eq:theta_gamma_dyn}.
We address the problem of designing the planning control $\pi$ and feedback controls $\omega$ and $z$ to minimize an integral cost of the predicted flow $\phi$ over the horizon $T$ such that the predicted flow $\phi(\cdot ; \bar x(t))$ and actual state $x(t)$ are in a prescribed safe set and the control $u(t) \in \SU$. 
We assume that $\SU, \Omega, \SZ$ are convex, $0 \in \Omega$, and $[0,1] \subseteq \SZ$.

Let $h_\rms \colon \BBR^n \to \BBR$ be continuously differentiable, and define the \textit{safe set}
\begin{equation}\label{eq:safe_set}
\SC_\rms \triangleq \left\{ x \in \BBR^n \colon h_\rms(x) \geq 0 \right\},
\end{equation}
which is nonempty and contains no isolated points. 
The set $\SC_\rms$ is not assumed to be control forward invariant with respect to \cref{eq:system} where $u$ is admissible.
The cost $J \colon \BBR^{\bar n} \to \BBR$ is 
\begin{equation*}
J(\bar{x}) \!\triangleq\! W\!\left(\phi(\gamma\!+\!T;\bar{x})\right) \!+\!\! \int_{\gamma}^{\gamma+T} \!\!\!\! R(\phi(\tau;\bar{x}),\!\pi(\tau,\phi(\tau;\bar{x});\theta))\rmd\tau,
\end{equation*}
where $W \colon \BBR^n \to \BBR$ and $R \colon \BBR^n \times \BBR^m \to \BBR$ are continuously differentiable.
For each real time $t \ge 0$, $J(\bar{x}(t))$ is a finite-horizon cost of the predicted flow $\phi(\cdot;\bar{x}(t))$ over a prediction horizon of length $T$.
The control objective is as follows.

\begin{problem}\rm\label{prob:main}
Design the planning control $\pi$ and feedback controls $(\omega(t) , z(t) ) \in \Omega \times \SZ$ such that for all $t\ge 0$, the receding horizon cost $J(\bar{x}(t))$ is minimized subject to:
\begin{enumerate}[leftmargin=0.9cm]
\renewcommand{\labelenumi}{(C\arabic{enumi})}
\renewcommand{\theenumi}{(C\arabic{enumi})}
\item\label{con:C1} For all $(t,\tau) \!\in\! [0,\infty) \times [\gamma(t),\gamma(t)+T]$, $\phi(\tau; \bar{x}(t)) \!\in \!\SC_\rms$.

\item\label{con:C2} For all $t \ge 0$, $\theta(t) \in \Theta$ and $\gamma(t) \geq 0$.

\end{enumerate}
\end{problem}

Next, consider a continuously differentiable \textit{backup control} $u_\rmb \colon \BBR^n \to \SU$, and define the associated closed-loop dynamics $f_\rmb(x) \triangleq f(x, u_\rmb(x))$.
Let $h_\rmb \colon \BBR^n \to \BBR$ be continuously differentiable, and define the \textit{backup safe set} $\SC_\rmb \triangleq \left\{ x \in \BBR^n \colon h_\rmb(x) \geq 0 \right\}$,
where $\SC_\rmb \subset \SC_\rms$.
Since $\SC_\rms$ is not control forward invariant, we make the following assumption. 

\begin{assumption}\rm\label{ass:backup}
There exists an extended class-$\SK$ function $\alpha_\rmb$ such that for all $x \in \SC_\rmb$, $L_{f_\rmb} h_\rmb(x) + \alpha_\rmb(h_\rmb(x)) \geq 0$.
\end{assumption}

Assumption~\ref{ass:backup} implies that $\SC_\rmb$ is forward invariant with respect to \eqref{eq:system} with $u = u_\rmb$.
If $\SC_\rmb$ is compact, then Assumption~\ref{ass:backup} is equivalent to forward invariance of $\SC_\rmb$ under $u = u_\rmb$ \cite{ames2016control}.
Assumption~\ref{ass:backup} is standard in the backup-CBF literature \cite{backupautomatic} and analogous to the terminal constraint in model predictive control with recursive feasibility \cite{mayne2000}. 
We also make the following technical assumption. 

\begin{assumption}\rm\label{ass:terminal}
There exists an extended class-$\SK$ function $\alpha_\rms$ such that for all $x \in \SC_\rmb$, $L_{f_\rmb} h_\rms(x) + \alpha_\rms(h_\rms(x)) \geq 0$.
\end{assumption}

If $\SC_\rmb$ is compact and $\SC_\rmb \subset \operatorname{int}\SC_\rms$, then Assumption~\ref{ass:terminal} is satisfied.
Specifically, if $\SC_\rmb$ is compact, then $M \triangleq \inf_{x \in \SC_\rmb} L_{f_\rmb} h_\rms(x)$ exists, and if $\SC_\rmb \subset \operatorname{int}\SC_\rms$, then there exists $\varepsilon > 0$ such that $h_\rms(x) \ge \varepsilon$ on $\SC_\rmb$. 
In this case, Assumption~\ref{ass:terminal} is satisfied for any $\alpha_\rms$ such that $\alpha_\rms(\varepsilon) \ge -M$.

Although $u_\rmb$ makes $\SC_\rmb$ forward invariant, $\SC_\rmb$ can be small relative to $\SC_\rms$.
Thus, it is not desirable to keep the predicted flow $\phi(\cdot ;\bar x(t))$ in $\SC_\rmb$ because this can lead to poor performance (i.e., large $J$). 
Instead, we construct 2 candidate P-CBFs, where the intersection of their zero-superlevel sets is the set of states $\bar x$ such that the predicted flow $\phi(\cdot ;\bar x)$ is in the safe set $\SC_\rms$ for the entire prediction horizon, and $\phi(\cdot ;\bar x)$ is in the backup safe set $\SC_\rmb$ at the terminal prediction time.

First, consider the minimum-over-horizon candidate P-CBF
\begin{equation}
\psi_\rms(\bar{x}) \triangleq \min\limits_{\tau \in [\gamma, \gamma + T]} h_{\rms}(\phi(\tau ; \bar{x})), \label{eq:psi_min}
\end{equation}
which is nonnegative if and only if $\phi(\tau;\bar{x}) \in \SC_\rms$ for all $\tau \in [\gamma,\gamma+T]$ (i.e., the entire prediction horizon). 
Next, consider the terminal-time candidate P-CBF
\begin{equation}
\psi_\rmt(\bar{x}) \triangleq h_\rmb(\phi(\gamma + T ; \bar{x})), \label{eq:psi_b}
\end{equation}
which is nonnegative if and only if $\phi(\gamma+T;\bar{x}) \in \SC_\rmb$.
Define
\begin{equation*} 
\Psi \triangleq \{\bar{x} \in \BBR^{\bar n} \colon k(\theta) \ge 0, \, \gamma \geq 0, \, \psi_\rms(\bar{x}) \ge 0, \, \psi_\rmt(\bar{x}) \ge 0 \},
\end{equation*}
which is the intersection of zero-superlevel sets associated with the candidate P-CBF pair $(\psi_\rms, \psi_\rmt)$. 
In other words, $\Psi$ is the set of $\bar x$ such that for all $\tau \in [\gamma,\gamma+T]$, the predicted flow satisfies $\phi(\tau;\bar{x}) \in \SC_\rms$; the predicted flow satisfies the terminal condition $\phi(\gamma+T;\bar{x}) \in \SC_\rmb$; and $(\theta,\gamma) \in \Theta \times [0,\infty)$.
Thus, \ref{con:C1} and \ref{con:C2} are satisfied if $\bar{x}(t) \in \Psi$ for all $t \geq 0$. 
Hence, we address Problem~\ref{prob:main} by solving the following problem.

\begin{problem}\rm\label{prob:reframed}
Design $\pi$ such that $(\psi_\rms, \psi_\rmt)$ is a P-CBF pair for \eqref{eq:augmented} on $\Psi$. 
Then, for all $t \geq 0$, determine $(\omega(t) , z(t) ) \in \Omega \times \SZ$ that minimize $J(\bar{x}(t))$ subject to $\bar x(t) \in \Psi$.
\end{problem}

The next result provides the directional derivative of $\psi_\rms$.
It is a direct consequence of \cite[Cor.~4.4, Thm.~5.1]{gauvin2009differential}.
The result is used to design $\pi$ such that $(\psi_\rms, \psi_\rmt)$ is a P-CBF pair.

\begin{lemma}\rm\label{prop:psi_lipschitz}
$\psi_\rms$ given by \eqref{eq:psi_min} is locally Lipschitz and directionally differentiable on $\Psi$.
Furthermore, for all $\bar{x} \in \Psi$,
\begin{equation*} 
D_\nu\psi_\rms(\bar{x}) = \min_{\tau \in \ST(\bar{x})} h_\rms'(\phi(\tau; \bar{x})) \left[ \frac{\partial \phi}{\partial \bar{x}}(\tau; \bar{x}) + \frac{\partial \phi}{\partial \tau}(\tau; \bar{x}) \frac{\partial \gamma}{\partial \bar{x}} \right]\nu,
\end{equation*}
where
\begin{equation}\label{eq:ST}
\ST(\bar{x}) \triangleq \operatorname*{argmin}_{\tau \in [\gamma, \gamma+T]} h_\rms(\phi(\tau; \bar{x})).
\end{equation}
\end{lemma}

\section{Safe Optimal Flow Control}\label{sec:sofc}

\subsection{Design of $\pi$ Such That $(\psi_\rms, \psi_\rmt)$ is a P-CBF Pair} 
\label{sec:design_pi_sub}

This subsection presents a construction for $\pi$ that smoothly transitions from a plan parameterized only by $\theta$ to the backup control $u_\rmb$ as the temporal argument of $\pi$ approaches $T$.

Consider the continuous \textit{parametric plan} $u_\rmp \colon [0,\infty) \times \BBR^d \to \BBR^m$, where for all $\tau \ge 0$, $u_\rmp(\tau;\cdot)$ is continuously differentiable on $\BBR^d$, and for all $(\tau,\theta) \in [0,\infty) \times \Theta$, $u_\rmp(\tau;\theta) \in \SU$.
Hence, $u_\rmp$ is admissible for all $\theta \in \Theta$.
One example is $u_\rmp=\kappa$, where $\kappa$ is given by Example~\ref{ex:foh}.

Let $\delta \in (0,T)$, and let $\xi \colon [0,\infty) \to [0,1]$ be continuously differentiable such that for all $\tau \in [0,T-\delta]$, $\xi(\tau) = 0$; for all $\tau \in [T,\infty)$, $\xi(\tau) = 1$; and $\xi$ is nondecreasing on $[T-\delta,T]$.
See \cite[Example~1]{safari2024TSCT} for an example of $\xi$.

Finally, let $\pi$ be given by 
\begin{equation}\label{eq:pi_def}
\pi(\tau, \phi; \theta) = \left [1 - \xi(\tau) \right ]  u_\rmp(\tau; \theta) + \xi(\tau) u_\rmb(\phi),
\end{equation}
and note that $\pi(\tau, \phi; \theta) = u_\rmp(\tau; \theta)$ for $\tau \in [0,T-\delta]$; $\pi(\tau, \phi; \theta) = u_\rmb(\phi)$ for $\tau \in [T,\infty)$; and for $\tau \in [T-\delta,T]$, $\pi$ transitions between $u_\rmp$ and $u_\rmb$, where $\delta \in (0,T)$ determines the transition rate. 
Since $\pi$ is a convex combination of $u_\rmp(\tau;\theta) \in \SU$ and $u_\rmb(\phi) \in \SU$, it follows that $\pi(\tau, \phi; \theta) \in \SU$.

To understand \eqref{eq:pi_def}, recall that the predicted flow \eqref{eq:flow_def} is evaluated under the planning control \eqref{eq:pi_def}, where $\tau \ge \gamma$. 
Thus, $\gamma$ determines how much of the prediction horizon $T$ is under $u_\rmp$ versus $u_\rmb$.
For small $\gamma$, the predicted flow $\phi$ is predominantly under $u_\rmp$. 
In contrast, for $\gamma > T-\delta$, the predicted flow $\phi$ is predominantly under $u_\rmb$.
Similarly, the executed control \eqref{eq:u_follows_plan} is $u = u_\rmp(\gamma; \theta)$ for $\gamma \le T-\delta$; and $u$ only transitions to $u_\rmb$ for $\gamma > T-\delta$.

For all $\bar{x} \in \Psi$, consider $K \colon \Psi \rightrightarrows \Omega \times \SZ$ defined by
\begin{align}\label{eq:K_constraint}
K(\bar{x}) \!\triangleq\! \{&\hat{v} \!\in\! \Omega \!\times\! \SZ \colon k^\prime(\theta)\hat{\omega} \!+\! \alpha_\theta(k(\theta)) \!\geq\! 0, \; \hat{z} \!+\! \alpha_\gamma(\gamma) \!\geq\! 0, \nn\\
&D_{\bar{f}(\bar{x}) + \bar{g}\hat{v}} \psi_\rms(\bar{x}) + \alpha_\rms(\psi_\rms(\bar{x})) \geq 0, \nn\\
&L_{\bar{f}} \psi_\rmt(\bar{x}) \!+\! L_{\bar{g}} \psi_\rmt(\bar{x}) \hat{v} \!+\! \alpha_\rmb(\psi_\rmt(\bar{x})) \!\geq\! 0\},
\end{align}
which is the constraint set for the candidate P-CBF pair $(\psi_\rms,\psi_\rmt)$. 
The next result shows $(\psi_\rms, \psi_\rmt)$ is a P-CBF pair.

\begin{theorem}\rm\label{thm:feasibility}
Let Assumptions~\ref{ass:backup} and~\ref{ass:terminal} be satisfied.
Then, the following statements hold:
\begin{enumerate}[leftmargin=0.7cm]
\item\label{thm:feasibility.a} $(\psi_\rms, \psi_\rmt)$ is a P-CBF pair for \eqref{eq:augmented} on $\Psi$.
\item\label{thm:feasibility.b} For all $\bar{x} \in \Psi$, $[0_{d \times 1} \;\; 1]^\top \in K(\bar{x})$.
\item\label{thm:feasibility.c} For all $\bar{x} \in \{\bar{x} \in \Psi \colon \gamma \geq T\}$, $0_{(d+1) \times 1} \in K(\bar{x})$.
\end{enumerate}
\end{theorem}

\begin{proof}
Let $\bar{x}_\rme \!=\! [x_\rme^\top \; \theta_\rme^\top \; \gamma_\rme]^\top \!\in\! \Psi$, $\omega_\rme \!=\! 0_{d \times 1}$, $z_\rme \!\in\! \{0, 1\}$, and $v_\rme \!\triangleq\! [\omega_\rme^\top \;\; z_\rme]^\top$.
Next, define
$c_1 \!\triangleq\! L_{\bar{f}} \psi_\rmt(\bar{x}_\rme) \!+\! L_{\bar{g}} \psi_\rmt(\bar{x}_\rme) v_\rme \!+\! \alpha_\rmb(\psi_\rmt(\bar{x}_\rme))$ and
$c_2 \!\triangleq\! D_{\bar{f}(\bar{x}_\rme) + \bar{g} v_\rme} \psi_\rms(\bar{x}_\rme) \!+\! \alpha_\rms(\psi_\rms(\bar{x}_\rme))$, and Lemma~\ref{prop:psi_lipschitz} implies that
\begin{align}
c_1 &= h_\rmb'(\phi(\gamma_\rme + T;\bar{x}_\rme)) \, \zeta(\gamma_\rme + T) + \alpha_\rmb(\psi_\rmt(\bar{x}_\rme)), \label{eq:c1_eta}\\
c_2 &= \min_{\tau \in \ST(\bar{x}_\rme)} h_\rms'(\phi(\tau;\bar{x}_\rme)) \, \zeta(\tau) + \alpha_\rms(\psi_\rms(\bar{x}_\rme)), \label{eq:c2_eta}
\end{align}
where
\begin{equation}\label{eq:zeta_eta_def}
\zeta(\tau) \triangleq \eta(\tau) + z_\rme \frac{\partial \phi}{\partial \tau}(\tau;\bar{x}_\rme), 
\end{equation}
and $\eta(\tau) \triangleq \frac{\partial \phi}{\partial \bar{x}}(\tau;\bar{x}_\rme) [\bar{f}(\bar{x}_\rme) + \bar{g} v_\rme]$, and substituting 
\eqref{eq:augmented.2}, \eqref{eq:augmented.3} and $\omega_\rme = 0$ yields 
\begin{equation}\label{eq:eta_expanded}
\eta(\tau) = \frac{\partial \phi}{\partial x}(\tau;\bar{x}_\rme) f(x_\rme, \pi(\gamma_\rme, x_\rme;\theta_\rme)) + \frac{\partial \phi}{\partial \gamma}(\tau;\bar{x}_\rme) z_\rme.
\end{equation}
Differentiating \eqref{eq:flow_def} with respect to $x$ and $\gamma$ yields
\begin{align}
&\frac{\partial \phi}{\partial x}(\tau;\bar{x}) = I + \int_{\gamma}^{\tau} A(\sigma;\bar{x}) \frac{\partial \phi}{\partial x}(\sigma;\bar{x}) \, \rmd\sigma, \label{eq:dphi_dx}\\
&\frac{\partial \phi}{\partial \gamma}(\tau;\bar{x}) \!=\! \!\int_{\gamma}^{\tau} \!\! \!A(\sigma;\bar{x}) \frac{\partial \phi}{\partial \gamma}(\sigma;\bar{x}) \, \rmd\sigma - f(x, \pi(\gamma, x;\theta)), \label{eq:dphi_dgamma}
\end{align}
where
\begin{align}
&A(\tau;\bar{x}) \triangleq \frac{\partial f}{\partial x}(\phi(\tau;\bar{x}), \pi(\tau, \phi(\tau;\bar{x}); \theta)) \nn\\
&\quad + \frac{\partial f}{\partial u}(\phi(\tau;\bar{x}), \pi(\tau, \phi(\tau;\bar{x}); \theta)) \xi(\tau) \frac{\partial u_\rmb}{\partial x}(\phi(\tau;\bar{x})).\label{eq:A_def}
\end{align}
Differentiating \eqref{eq:eta_expanded} with respect to $\tau$ and substituting \cref{eq:dphi_dx,eq:dphi_dgamma} yields
\begin{equation}\label{eq:eta_ODE}
\frac{\rmd \eta}{\rmd \tau}(\tau) \!=\! A(\tau;\bar{x}_\rme) \eta(\tau), \,\,\,\, \eta(\gamma_\rme) \!=\! (1 - z_\rme) f(x_\rme, \pi(\gamma_\rme, x_\rme;\theta_\rme)).
\end{equation}

To prove \ref{thm:feasibility.b}, let $z_\rme = 1$. 
Thus, \eqref{eq:eta_ODE} implies that for all $\tau \ge \gamma_\rme$, $\eta(\tau) = 0$, which combined with \eqref{eq:zeta_eta_def} implies that 
\begin{equation}\label{eq:eta.2}
\zeta(\tau) = \frac{\partial \phi}{\partial \tau}(\tau;\bar{x}_\rme).    
\end{equation}

Since $\gamma_\rme + T \geq T$, it follows from \cref{eq:pi_def,eq:flow_velocity,eq:eta.2} that $\zeta(\gamma_\rme+T)= \frac{\partial \phi}{\partial \tau}(\gamma_\rme+T;\bar{x}_\rme) = f_\rmb(\phi(\gamma_\rme+T;\bar{x}_\rme))$. 
Thus, \eqref{eq:c1_eta} becomes $c_1 = L_{f_\rmb} h_\rmb(\phi(\gamma_\rme+T;\bar{x}_\rme)) + \alpha_\rmb(h_\rmb(\phi(\gamma_\rme+T;\bar{x}_\rme)))$.
Since, in addition, $\bar{x}_\rme \in \Psi$ implies that $\phi(\gamma_\rme+T;\bar{x}_\rme) \in \SC_\rmb$, it follows from Assumption~\ref{ass:backup} that $c_1 \geq 0$.

Next, using \eqref{eq:eta.2} in \eqref{eq:c2_eta} implies that
\begin{equation*}
c_2 = \min_{\tau \in \ST(\bar{x}_\rme)} \frac{\rmd}{\rmd \tau} h_\rms(\phi(\tau;\bar{x}_\rme)) + \alpha_\rms(\psi_\rms(\bar{x}_\rme)).
\end{equation*}
Let $\tau_* \in \ST(\bar{x}_\rme)$, and consider 3 cases: $\tau_* \in (\gamma_\rme, \gamma_\rme + T)$, $\tau_* = \gamma_\rme$, and $\tau_* = \gamma_\rme + T$.
First, consider $\tau_* \in (\gamma_\rme, \gamma_\rme + T)$, which implies that $\tau_*$ is an interior minimizer of $h_\rms(\phi(\tau;\bar{x}_\rme))$.
Thus, $\frac{\rmd}{\rmd \tau} h_\rms(\phi(\tau_*;\bar{x}_\rme)) = 0$, and since $\bar{x}_\rme \in \Psi$, it follows that $c_2 = \alpha_\rms(\psi_\rms(\bar{x}_\rme)) \geq 0$.
Next, consider $\tau_* = \gamma_\rme$, which implies that $\gamma_\rme$ is a minimizer of $h_\rms(\phi(\tau;\bar{x}_\rme))$ at the left endpoint of $[\gamma_\rme, \gamma_\rme+T]$.
Thus, $\frac{\rmd}{\rmd \tau} h_\rms(\phi(\gamma_\rme;\bar{x}_\rme)) \geq 0$, and it follows that $c_2 \geq \alpha_\rms(\psi_\rms(\bar{x}_\rme)) \geq 0$.
Next, consider $\tau_* = \gamma_\rme + T$.
Since $\gamma_\rme + T \geq T$, it follows from \cref{eq:pi_def,eq:flow_velocity} that $\frac{\rmd}{\rmd \tau} h_\rms(\phi(\gamma_\rme+T;\bar{x}_\rme)) = L_{f_\rmb} h_\rms(\phi(\gamma_\rme+T;\bar{x}_\rme))$ and $\psi_\rms(\bar{x}_\rme) = h_\rms(\phi(\gamma_\rme+T;\bar{x}_\rme))$.
Since, in addition, $\bar{x}_\rme \in \Psi$ implies $\phi(\gamma_\rme+T;\bar{x}_\rme) \in \SC_\rmb$, it follows from Assumption~\ref{ass:terminal} that $c_2 \geq 0$.
Since $\bar{x}_\rme \in \Psi$, it follows that $k(\theta_\rme) \geq 0$ and $\gamma_\rme \geq 0$, which implies that $\alpha_\theta(k(\theta_\rme)) \geq 0$ and $z_\rme + \alpha_\gamma(\gamma_\rme) \geq 0$.
Since, in addition, $c_1, c_2 \geq 0$, \cref{eq:K_constraint} implies that $v_\rme \in K(\bar{x}_\rme)$, which proves \ref{thm:feasibility.b}.

To prove \ref{thm:feasibility.c}, let $\bar{x}_\rme \!\in\! \{\bar{x} \!\in\! \Psi \colon \gamma \!\geq\! T\}$ and $z_\rme \!=\! 0$.
Since $\gamma_\rme \!\geq\! T$, \eqref{eq:pi_def} implies that for all $\tau \!\ge\! \gamma_\rme$, $\xi(\tau) \!=\! 1$ and $\pi(\tau, \phi(\tau;\bar{x}_\rme); \theta_\rme) \!=\! u_\rmb(\phi(\tau;\bar{x}_\rme))$.
Thus, \cref{eq:flow_velocity,eq:A_def} imply that $\frac{\partial \phi}{\partial \tau}(\tau;\bar{x}_\rme) \!=\! f_\rmb(\phi(\tau;\bar{x}_\rme))$ and $A(\tau;\bar{x}_\rme) \!=\! f_\rmb'(\phi(\tau;\bar{x}_\rme))$.
Since, in addition, $\frac{\rmd}{\rmd \tau} f_\rmb(\phi(\tau;\bar{x}_\rme)) \!=\! f_\rmb'(\phi(\tau;\bar{x}_\rme)) f_\rmb(\phi(\tau;\bar{x}_\rme))$ and $z_\rme \!=\! 0$, it follows that $\eta \!=\! f_\rmb(\phi(\cdot;\bar{x}_\rme))$ satisfies \eqref{eq:eta_ODE}.
Thus, for all $\tau \!\ge\! \gamma_\rme$, $\eta(\tau) \!=\! \frac{\partial \phi}{\partial \tau}(\tau;\bar{x}_\rme)$, which combined with \eqref{eq:zeta_eta_def} yields \eqref{eq:eta.2}.
The remainder of the proof of \ref{thm:feasibility.c} is identical to the last two paragraphs of the proof of~\ref{thm:feasibility.b}.

Finally, \ref{thm:feasibility.b} and Definition~\ref{def:pcbf} imply that $(\psi_\rms, \psi_\rmt)$ is a P-CBF pair, which confirms~\ref{thm:feasibility.a}.
\end{proof}

The following result shows that any control selected pointwise from $K(\bar{x})$ makes $\Psi$ forward invariant.
The result is an immediate consequence of Theorems~\ref{thm:pcbf_invariance} and~\ref{thm:feasibility}.

\begin{corollary}\rm\label{cor:forward_invariance}
Let Assumptions~\ref{ass:backup} and~\ref{ass:terminal} be satisfied, and let $v_{\rm fi} \colon \Psi \to \Omega \times \SZ$ be such that for all $\bar{x} \in \Psi$, $v_{\rm fi}(\bar{x}) \in K(\bar{x})$.
Let $\bar{x}_0 \in \Psi$, and assume there exists $t_\rmm \in (0, \infty]$ such that \eqref{eq:augmented} with $v = v_{\rm fi}$ has a unique solution on $[0, t_\rmm)$.
Then, for all $t \in [0, t_\rmm)$, $\bar{x}(t) \in \Psi$.
\end{corollary}

\subsection{Safe Optimal Controls $\omega$ and $z$}
\label{sec:optimal_control}

Since $J$ can be nonlinear and nonconvex, Problem~\ref{prob:reframed} cannot generally be solved with a convex optimization. 
However, $\frac{\rmd J}{\rmd t}$ along the trajectories of \eqref{eq:augmented} is affine in the control $v$.  
To make $\frac{\rmd J}{\rmd t}$ small, consider the quadratic cost
\begin{align}\label{eq:qp_cost}
\SJ(\hat{v}; \bar{x}) &\triangleq \frac{\partial J}{\partial \theta} \hat{\omega} + \hat{\omega}^\top Q_\omega \hat{\omega} + q_z \hat{z}^2 + \lambda \hat{z},
\end{align}
where $Q_\omega \in \BBR^{d \times d}$ is positive definite, $q_z > 0$, and $\lambda \geq 0$.
The first term of \eqref{eq:qp_cost} makes $\frac{\rmd J}{\rmd t}$ small, while $\hat{\omega}^\top Q_\omega \hat{\omega}$ and $q_z \hat{z}^2$ provide regularization to make \eqref{eq:qp_cost} strictly convex.
The term $\lambda \hat{z}$ penalizes increasing $\gamma$, which encourages the parametric plan over the backup controller.
The next result shows $\SJ$ has a unique minimizer on $K(\bar{x})$.
See \cite[Prop. 5]{safari2026flowbarrier} for a proof.

\begin{proposition}\rm\label{prop:K_convex}
Let Assumptions~\ref{ass:backup} and~\ref{ass:terminal} be satisfied.
Then, for all $\bar{x} \in \Psi$, $K(\bar{x})$ is convex, and $\operatorname*{argmin}_{\hat{v} \in K(\bar{x})} \SJ(\hat{v}; \bar{x})$ exists and is unique.
\end{proposition}

For all $\bar{x} \in \Psi$, define the safe optimal flow control
\begin{equation}\label{eq:v_star_def}
v_*(\bar x) \triangleq \operatorname*{argmin}_{\hat{v} \in K(\bar{x})} \SJ(\hat{v}; \bar{x}),
\end{equation}
where Proposition~\ref{prop:K_convex} implies that $v_*(\bar{x})$ exists and is unique.

The next corollary is the main result of this article.
The result shows that $v_*$ makes $\Psi$ forward invariant.

\begin{corollary}\rm\label{cor:main_safety}
Let Assumptions~\ref{ass:backup} and~\ref{ass:terminal} be satisfied.
Let $\bar{x}_0 \in \Psi$, and assume there exists $t_\rmm \in (0, \infty]$ such that \eqref{eq:augmented} with $v = v_*$ has a unique solution on $[0, t_\rmm)$.
Then, for all $t \in [0, t_\rmm)$, $\bar{x}(t) \in \Psi$.
Furthermore, if $\gamma_0 \!\in\! [0,T]$, then for all $t \!\in\! [0, t_\rmm)$, $\!\gamma(t) \!\in\! [0,T]$.
\end{corollary}

\begin{proof}
Corollary~\ref{cor:forward_invariance} implies that for all $t \in [0, t_\rmm)$, $\bar{x}(t) \in \Psi$.

To prove the last sentence, let $\gamma_0 \in [0,T]$, and assume for contradiction that there exists $t_1 \!\in\! [0, t_\rmm)$ such that $\gamma(t_1)\! = \! T$ and $z(t_1) \!>\!0$.
Since $\pi(\tau, \cdot; \theta) = u_\rmb$ for all $\tau \ge \gamma(t_1)$, it follows that $\frac{\partial J}{\partial \theta}(\bar x_1) = 0$, where $\bar x_1 \triangleq \bar x(t_1)$.
Since, in addition, $v_*(\bar x_1) = [ \omega(t_1)^\top \quad z(t_1) ]^\top$ and $z(t_1) >0$, it follows from \eqref{eq:qp_cost} that $\SJ(v_*(\bar x_1); \bar x_1) > 0 = \SJ(0;\bar x_1)$, which is a contradiction because $0 \in K(\bar{x}_1)$ from Theorem~\ref{thm:feasibility}. 
\end{proof}

Corollary~\ref{cor:main_safety} shows that the control $v_*$ guarantees that the predicted flow $\phi(\cdot;\bar{x}(t))$ and actual state $x(t)$ are in the safe set $\SC_\rms$ for all time $t \in [0,t_\rmm)$, while minimizing $\SJ$, which aims to decrease the cost $J$ along the trajectories of \eqref{eq:augmented}.
If $\gamma_0 \in [0,T]$, then $\gamma$ stays in $[0,T]$.
In addition, the cost \eqref{eq:qp_cost} incentivizes keeping $\gamma$ small to keep the majority of the prediction window under $u_\rmp$ rather than $u_\rmb$.

\section{QP Implementation}\label{sec:qp}

This section presents \textit{FlowBarrier}, which is the QP implementation of \eqref{eq:v_star_def} in the case where $\Omega$ and $\SZ$ are convex polyhedra.
Lemma~\ref{prop:psi_lipschitz} implies that the constraint on $\psi_\rms$ in \eqref{eq:K_constraint} is equivalent to a family of affine constraints.
However, $\ST(\bar{x})$ may contain infinitely many points, which makes \eqref{eq:v_star_def} a semi-infinite QP.
We use discretization to solve \eqref{eq:v_star_def} although other approaches can be used \cite{lopez2007semi}. 
Let $N$ be a positive integer, and define $\tau_i \!\triangleq\! \gamma + iT/N$ for $i \in \{0,1,\ldots,N\}$.
Define $\ST_\rme(\bar{x}) \triangleq \operatorname*{argmin}_{\hat\tau \in \{\tau_0, \ldots, \tau_N\}} h_\rms(\phi(\hat\tau; \bar{x}))$, which contains the discrete prediction times at which $h_\rms(\phi(\,\cdot\,;\bar{x}))$ attains its minimum.
Next, let $K_\rme(\bar{x})$ be given by \eqref{eq:K_constraint}, where $D_{\bar{f}(\bar{x}) + \bar{g}\hat{v}} \psi_\rms(\bar{x})$ is replaced by
\begin{equation*}
h_\rms'(\phi(\tau; \bar{x})) \!\left[ \frac{\partial \phi}{\partial \bar{x}}(\tau; \bar{x}) + \frac{\partial \phi}{\partial \tau}(\tau; \bar{x}) \frac{\partial \gamma}{\partial \bar{x}} \right]\! [\bar{f}(\bar{x}) + \bar{g}\hat{v}]
\end{equation*}
for all $\tau \in \ST_\rme(\bar{x})$.
Thus, $K_\rme(\bar{x})$ has a finite number of affine constraints, and \eqref{eq:v_star_def} is approximated by
\begin{equation}\label{eq:primal_qp}
v_{* \rme}(\bar{x}) \triangleq \operatorname*{argmin}_{\hat{v} \in K_\rme(\bar{x})} \; \SJ(\hat{v}; \bar{x}).
\end{equation}
This QP requires $\frac{\partial J}{\partial \theta}$, $\frac{\partial \psi_\rmt}{\partial \bar{x}}$, and $\frac{\partial h_\rms(\phi(\tau;\bar{x}))}{\partial \bar{x}}$, which can be computed efficiently using the adjoint approach \cite{chen2018neural,safari2026flowbarrier}.

In this article, $\SC_\rms$, $\SC_\rmb$, and $\Theta$ are each the zero-superlevel set of a single function.
It can be useful to define each set as the intersection of zero-superlevel sets of multiple functions.
The log-sum-exponential soft minimum \cite{backupautomatic,safari2024TSCT} can be used to compose multiple barrier functions into a single one.
Let $\rho \!>\! 0$, and consider $\operatorname{softmin}_\rho \colon \BBR^{n_{\rm sm}} \!\to\! \BBR$ defined by $\operatorname{softmin}_\rho(z_1,\!\ldots,\!z_{n_{\rm sm}}) \!\triangleq\! -(1/\rho) \ln \sum_{i=1}^{n_{\rm sm}} e^{-\rho z_i}$, which is a continuously differentiable lower bound on the minimum with worst-case conservativeness $(\ln n_{\rm sm})/\rho$.
To illustrate a composite construction of $h_\rms$, consider $n_\rms$ barrier functions $b_1,\!\ldots,\!b_{n_\rms}$.
Then, we let the safe set $\SC_\rms$ be \eqref{eq:safe_set}, where $h_\rms(x) \!=\! \operatorname{softmin}_{\rho_\rms}(b_1(x),\!\ldots,\!b_{n_\rms}(x))$.
A similar construction can be used for $h_\rmb$.
See \cite{rabiee2024closed,safari2024ACC,safari2025safe} for more information.

The soft minimum can also be used to construct $k$. 
For example, consider $u_\rmp=\kappa$, where $\kappa$ is given by Example~\ref{ex:foh}, and $\SU = \{u \in \BBR^m \colon a_j^\top u + d_j \geq 0, \, j = 1,\ldots,r\}$ is a convex polyhedron.
Since $u_\rmp$ is a convex combination of $\theta_1,\ldots,\theta_p$, the admissible parameter set $\Theta$ can be constructed with
\begin{align}\label{eq:k_construction}
k(\theta) = \operatorname{softmin}_{\rho_k}(& a_1^\top \theta_1 + d_1, \ldots, a_r^\top \theta_1 + d_r, \ldots, \nn\\
&a_1^\top \theta_p + d_1, \ldots, a_r^\top \theta_p + d_r),
\end{align}
which ensures $\Theta \subseteq \SU^p$ and $\Theta \to \SU^p$ as $\rho_k \to \infty$.

\section{Application to Autonomous Car}\label{sec:simulation}

Consider the car modeled by \eqref{eq:system}, where
\begin{equation*}
    f(x,u) = \begin{bmatrix}
     \nu\cos\vartheta \\
     \nu\sin\vartheta \\
     u_1 \\
     \nu \tan u_2
    \end{bmatrix},
    \,
    x = \begin{bmatrix}
    q_\rmx\\
    q_\rmy\\
    \nu\\
    \vartheta
    \end{bmatrix},
    \,
    u = \begin{bmatrix}
    u_1\\
    u_2
    \end{bmatrix},
\end{equation*}
where $q \triangleq [ \, q_\rmx \quad q_\rmy \, ]^\top$ is the position, $\nu$ is the speed, $\vartheta$ is the heading angle, $u_1$ is the acceleration, and $u_2$ is the steering angle.
Note that $f$ is not affine in $u$.
Let $\bar u_1 = 2$, $\bar u_2 = 1$, and $\SU = \{ u \in \BBR^2 \colon u_1 \in [-\bar u_1, \bar u_1], \, u_2 \in [-\bar u_2, \bar u_2] \}.$

Consider the map shown in Figure~\ref{fig:trajectory}, which has $46$ circular obstacles and a wall.
For $i \in \{1, \ldots, 46\}$, the area outside the $i$th obstacle is the zero-superlevel set of $b_i(x) = \|q - c_i\| - r_i$,
where $c_i \in \BBR^2$ and $r_i > 0$ are the center and radius.
The area inside the wall is the zero-superlevel set of $b_{47}(x) = 10 - (q_\rmx^{20} + q_\rmy^{20})^{1/20}$.
The speed bounds are $b_{48}(x) = 2 - \nu$ and $b_{49}(x) = \nu + 2$.
Let $\rho_\rms = 20$, $h_\rms(x) = \operatorname{softmin}_{\rho_\rms}(b_1(x), \ldots, b_{49}(x))$, and $h_\rmb(x) = h_\rms(x) - 0.5 {\nu^2}/{\bar u_1} - \epsilon_\rmb$, where $\epsilon_\rmb > 0$.

The objective is for the car to move to the desired state $x_\rmd \in \BBR^4$ without violating constraints. 
Thus, we consider the cost $J$, where $R(x,u) = W(x) = \|x - x_\rmd\|^2$.

The parametric plan $u_\rmp$ is given by Example~\ref{ex:foh} with degree-one B-splines and $\xi$ given by \cite[Example~1]{safari2024TSCT} with $r = 1$ and $\delta = 0.2T$. 
The backup control is $u_\rmb(x) = [-\bar u_1 \tanh(15\nu) \;\; 0]^\top$, which satisfies Assumption~\ref{ass:backup} with any extended class-$\SK$ function $\alpha_\rmb$. 
The set $\Theta$ is constructed using \eqref{eq:k_construction} with $\rho_k = 50$. 
Let $\Omega = \BBR^d$ and $\SZ = (-\infty, 1]$.

We implement FlowBarrier \cref{eq:theta_gamma_dyn,eq:primal_qp,eq:qp_cost,eq:pi_def,eq:u_follows_plan} with $\alpha_\rms(\psi_\rms) = 12\psi_\rms$, $\alpha_\rmb(\psi_\rmt) = 5\psi_\rmt$, $\alpha_\theta(k) = 20k$, $\alpha_\gamma(\gamma) = 0.1\gamma$, $Q_\omega = 30 I_d$, $q_z = 10^{-6}$, $\lambda = 10^3$, and $N = 80$.
We use 2 configurations: (a) $T = 4\,\rms$, $p = 80$, $100\,\text{Hz}$ update; and (b) $T = 6\,\rms$, $p = 150$, $100\,\text{Hz}$ update.

For comparison, we implement nonlinear model predictive control (NMPC) \cite{mayne2000}, which seeks to minimize the $N$-point sampling of the cost $J$ subject to $u \in \SU$ and $h_\rms(x) \geq 0$ at each step of the prediction horizon and the terminal constraint $h_\rmb(x(T)) \geq 0$. 
We use 2 configurations: (c) $T = 4\,\rms$, $N = 80$, $100\,\text{Hz}$ update; and (d) $T = 4\,\rms$, $N = 80$, $20\,\text{Hz}$ update.

To compare methods, we conduct $100$ different tasks---$10$ initial states uniformly distributed along the bottom of the map and $10$ desired states uniformly distributed along the top. 
Each simulation is run for $20\,\rms$.
Trajectories are categorized as \emph{reached} if $\|x - x_\rmd\| \leq 0.5$ within $20\,\rms$, or \emph{stuck} if they do not reach the goal.
Table~\ref{tab:results} summarizes the results, where $J_{\rm cum} \triangleq \int_0^{20} R(x(t),u(t)) \, \rmd t$.
Cases (a) and (c) compare FlowBarrier and NMPC with the same horizon $T=4$, discretization $N=p=80$, and update rate $100$ Hz. 
For this configuration, FlowBarrier and NMPC achieve no constraint violations with comparable success rates. 
The cost $J_{\rm cum}$ is slightly lower with NMPC, but the computation time is $30\times$ longer than with FlowBarrier. 
In fact, the NMPC computation time is too large for the 100~Hz real-time implementation on the processor used. 
In contrast, FlowBarrier is computationally efficient because it propagates the predicted flow using forward integration, computes gradients using the adjoint method, and solves a QP. 
All of this takes approximately 1~ms on the processor used. 
This low computation time allows for longer prediction horizons and/or a more complex parameterization of $u_\rmp$.
Case (b) shows that these changes can improve performance and success rate.
In fact, FlowBarrier (b) achieves the highest success rate and comparable $J_{\rm cum}$ and is still $10\times$ faster than NMPC. 
Case (d) uses a 20-Hz update, which is near-real-time for NMPC, and results in a slight degradation of performance. 
For higher dimensional systems, NMPC can become intractable for real time. 
In contrast, FlowBarrier scales linearly with state dimension. 
Figure~\ref{fig:trajectory} shows the trajectories for cases (b) and (d).

\begin{figure}[t!]
\centering
\includegraphics[width=\columnwidth]{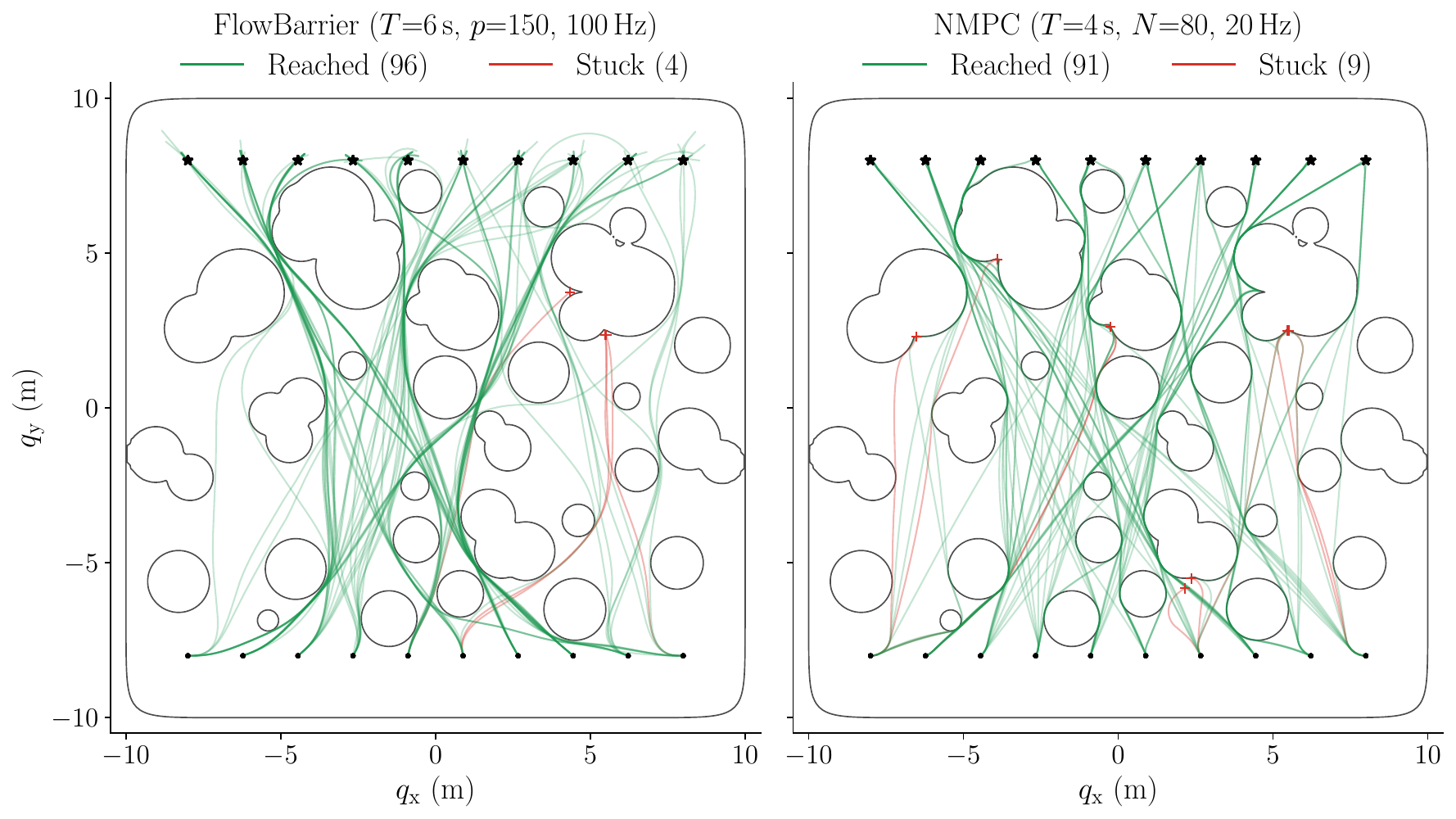}
\caption{Trajectories for \!(b)\! and \!(d), where reached is green and stuck is red.}
\label{fig:trajectory}

\captionof{table}{Comparison of FlowBarrier and NMPC across 100 trials.}
\label{tab:results}
\resizebox{\columnwidth}{!}{%
\scriptsize
\setlength{\tabcolsep}{4pt}
\begin{tabular}{@{}lcccc@{}}
\hline
Method & Success & Safety & $J_{\rm cum}$ & Comp.\ time (ms) \\
\hline
(a) FlowBarrier ($T\!=\!4$, $p\!=\!80$, 100\,Hz) & $89\%$ & $\mathbf{100\%}$ & $14.9 \pm 7.0$ & $\mathbf{[0.76,\, 1.10]}$ \\
(b) FlowBarrier ($T\!=\!6$, $p\!=\!150$, 100\,Hz) & $\mathbf{96\%}$ & $\mathbf{100\%}$ & $13.9 \pm 3.8$ & $[1.23,\, 2.11]$ \\
(c) NMPC ($T\!=\!4$, $N\!=\!80$, 100\,Hz) & $92\%$ & $\mathbf{100\%}$ & $\mathbf{12.2 \pm 5.3}$ & $[18.9,\, 57.0]$ \\
(d) NMPC ($T\!=\!4$, $N\!=\!80$, 20\,Hz) & $91\%$ & $\mathbf{100\%}$ & $12.5 \pm 5.4$ & $[18.8,\, 60.9]$ \\
\hline
\end{tabular}}
\end{figure}

\bibliographystyle{ieeetr}
\bibliography{SafeFlow}

@article{safari2026flowbarrier,
  title={Predicted-Flow Control Barrier Functions for Real-Time Safe Optimal Control},
  author={Safari, Amirsaeid and Hoagg, Jesse B},
  journal={arXiv:2606.00297},
  year={2026}
}

@article{safari2024TSCT,
  title={Time-Varying Soft-Maximum Barrier Functions for Safety in Unmapped and Dynamic Environments},
  author={Safari, Amirsaeid and Hoagg, Jesse B.},
  journal={IEEE Trans. Contr. Syst. Technol.},
  year={2025}
}

@inproceedings{safari2024ACC,
  title={Time-varying soft-maximum control barrier functions for safety in an a priori unknown environment},
  author={Safari, Amirsaeid and Hoagg, Jesse B},
  booktitle={Proc. Amer. Contr. Conf.},
  year={2024}
}

@article{ames2016control,
  title={Control barrier function based quadratic programs for safety critical systems},
  author={Ames, Aaron D and others},
  journal={IEEE Trans. Autom. Contr.},
  year={2017}
}

@article{backupautomatic,
  title={Soft-minimum and soft-maximum barrier functions for safety with actuation constraints},
  author={Rabiee, Pedram and Hoagg, Jesse B.},
  journal={Automatica},
  year={2025}
}

@article{rabiee2024closed,
  title={A Closed-Form Control for Safety Under Input Constraints Using a Composition of Control Barrier Functions},
  author={Rabiee, Pedram and Hoagg, Jesse B},
  journal={arXiv:2406.16874},
  year={2024}
}

@article{wieland2007,
  title={Constructive safety using control barrier functions},
  author={Wieland, Peter and Allg{\"o}wer, Frank},
  journal={IFAC Proc.},
  year={2007}
}

@inproceedings{wang2018,
  title={Permissive Barrier Certificates for Safe Stabilization Using Sum-of-Squares},
  author={Wang, Li and others},
  booktitle={Proc. Amer. Contr. Conf.},
  year={2018}
}

@article{gurriet2020,
  title={A Scalable Safety Critical Control Framework for Nonlinear Systems},
  author={Gurriet, Thomas and others},
  journal={IEEE Access},
  year={2020}
}

@book{khalil2002nonlinear,
  title={Nonlinear Systems},
  author={Khalil, Hassan K},
  edition={3rd},
  publisher={Prentice Hall},
  year={2002}
}

@inproceedings{gauvin2009differential,
  title={Differential properties of the marginal function in mathematical programming},
  author={Gauvin, Jacques and others},
  booktitle={Optim. Stab. Math. Program.},
  year={2009}
}

@article{chen2018neural,
  title={Neural ordinary differential equations},
  author={Chen, Ricky TQ and others},
  journal={Adv. Neural Inf. Process. Syst.},
  year={2018}
}

@article{blanchini1999,
  title={Set invariance in control},
  author={Blanchini, Franco},
  journal={Automatica},
  year={1999}
}

@article{mayne2000,
  title={Constrained model predictive control: Stability and optimality},
  author={Mayne, David Q and others},
  journal={Automatica},
  year={2000}
}

@inproceedings{choi2021hj,
  title={Robust control barrier--value functions for safety-critical control},
  author={Choi, Jason J. and others},
  booktitle={IEEE Conf. Dec. Contr.},
  year={2021}
}

@inproceedings{breeden2022,
  title={Predictive control barrier functions for online safety critical control},
  author={Breeden, Joseph and Panagou, Dimitra},
  booktitle={IEEE Conf. Dec. Contr.},
  year={2022}
}

@article{agrawal2025gatekeeper,
  title={The gatekeeper: Online safety verification and control for nonlinear systems in dynamic environments},
  author={Agrawal, Devansh R. and Panagou, Dimitra},
  journal={IEEE Trans. Robot.},
  year={2025}
}

@inproceedings{vahs2024,
  title={Forward invariance in trajectory spaces for safety-critical control},
  author={Vahs, Matti and others},
  booktitle={IEEE Int. Conf. Robot. Autom.},
  year={2025}
}

@inproceedings{grandia2021layered,
  title={Multi-layered safety for legged robots via control barrier functions and model predictive control},
  author={Grandia, Ruben and others},
  booktitle={IEEE Int. Conf. Robot. Autom.},
  year={2021}
}

@article{lopez2007semi,
  title={Semi-infinite programming},
  author={L{\'o}pez, Marco and others},
  journal={Eur. J. Oper. Res.},
  year={2007}
}

@inproceedings{singletary2022safe,
  title={Safe drone flight with time-varying backup controllers},
  author={Singletary, Andrew and others},
  booktitle={IEEE/RSJ Int. Conf. Intell. Robot. Syst.},
  year={2022}
}

@inproceedings{sforni2024receding,
  title={Receding horizon CBF-based multi-layer controllers for safe trajectory generation},
  author={Sforni, Lorenzo and others},
  booktitle={Proc. Amer. Contr. Conf.},
  year={2024}
}

@inproceedings{safari2025safe,
  title={Safe Navigation in Unmapped Environments for Robotic Systems with Input Constraints},
  author={Safari, Amirsaeid and Hoagg, Jesse B},
  booktitle={IEEE Conf. Dec. Contr.},
  year={2025}
}

\end{document}